\documentclass[conference]{IEEEtran}

\usepackage{booktabs}       %
\usepackage{enumitem}
\usepackage{pgfplots}
\usepackage{xfrac}             %
\usepackage[colorlinks]{hyperref}
\usepackage{microtype}      %
\usepackage[textsize=tiny]{todonotes}
\usepackage{calc}
\usepackage{cite}
\usepackage{amsmath,amssymb,amsfonts,amsthm}
\usepackage{graphicx}
\usepackage{textcomp}
\usepackage{xcolor}
\usepackage{import}
\usepackage{cleveref}
\usepackage{tikzscale}
\usepackage{graphicx}
\usepackage{tikz}

\hypersetup{
    colorlinks=true,
    linkcolor=blue,
    filecolor=magenta,
    urlcolor=blue,
    citecolor = blue,
}

\everymath=\expandafter{\the\everymath\displaystyle}

\usepackage{algorithm}
\usepackage{algpseudocode}

\algblockdefx{MRepeat}{EndRepeat}{\textbf{repeat}}{}
\algnotext{EndRepeat}

\newcommand{\cvc}{\textsc{cvc5}}

\newcommand{\definecount}[2]{\ensuremath{\mathsf{Count{#1}_{\downarrow {#2}}}}}

\newcommand{\ctp}{\definecount{\mathcal{T}}{\mathcal{P}}}

\newcommand{\cfplrabv}{\definecount{BVFPLRA}{BV}}

\newcommand{\cbvfp}{\definecount{BVFP}{BV}}
\newcommand{\mparagraph}[1]{\par \vspace{.1cm} \noindent \textit{#1}}

\newcommand{\est}[0]{\ensuremath{\mathsf{Est}}}

\newcommand{\vars}[1]{\ensuremath{\mathsf{Vars}(#1)}}
\newcommand{\sol}[0]{\ensuremath{\mathsf{Sol}}}

\newcommand{\smap}[0]{\ensuremath{\mathsf{SMTApproxMC}}}
\newcommand{\tool}[0]{\ensuremath{\mathsf{pact}}}
\newcommand{\pact}[0]{\tool}

\newcommand{\enum}[0]{\ensuremath{\mathsf{enum}}}

\newcommand{\saturatingcounter}[0]{\ensuremath{\mathsf{SaturatingCounter}}}
\newcommand{\nextIndex}[0]{\ensuremath{\mathsf{NextIndex}}}
\newcommand{\getconstants}[0]{\ensuremath{\mathsf{GetConstants}}}

\newcommand{\TT}{\ensuremath{\mathcal{T}}}
\newcommand{\PP}{\ensuremath{\mathcal{P}}}

\newcommand{\tot}[0]{3,119}
\newcommand{\mxb}[0]{83}
\newcommand{\mxp}[0]{456}
\newcommand{\mxpp}[0]{33}
\newcommand{\mxps}[0]{40}
\newcommand{\mxba}{64}

\newcommand{\sota}[0]{state-of-the-art}

\newcommand{\ope}{\ensuremath{1+\varepsilon}}
\newcommand{\solproj}[2]{\ensuremath{\mathsf{Sol}(#1)_{\downarrow{#2}}}}

\everymath{\textstyle}

\newcommand{\killthis}[1]{}

\newcommand{\HH}{\ensuremath{\mathcal{H}}}
\newcommand{\OO}[1]{\ensuremath{\mathcal{O}{(#1)}}}
\newcommand{\hshift}{\ensuremath{\mathcal{H}_{\textit{shift}}}}
\newcommand{\hxor}{\ensuremath{\mathcal{H}_{\textit{xor}}}}
\newcommand{\hprime}{\ensuremath{\mathcal{H}_{\textit{prime}}}}
\newcommand{\pshift}{\ensuremath{\pact_{\textit{shift}}}}
\newcommand{\pxor}{\ensuremath{\pact_{\textit{xor}}}}
\newcommand{\pprime}{\ensuremath{\pact_{\textit{prime}}}}

\newcommand{\CDM}{\ensuremath{\mathsf{CDM}}}

\newcommand{\ceil}[1]{\ensuremath{\lceil #1 \rceil}}
\newcommand{\floor}[1]{\ensuremath{\lfloor #1 \rfloor}}

\newcommand{\thresh}[0]{\ensuremath{\mathsf{thresh}}}
\newcommand{\itercount}[0]{\ensuremath{\mathsf{numIt}}}

\newcommand{\FindMedian}{\ensuremath{\mathsf{FindMedian}}}

\newcommand{\generatehash}[0]{\ensuremath{\mathsf{GenerateHash}}}
\newcommand{\GetCount}[0]{\ensuremath{\mathsf{GetCount}}}
\newcommand{\iterdone}[0]{\ensuremath{\mathsf{countFound}}}

\newcommand{\iters}[0]{\ensuremath{\mathsf{iters}}}

\newcommand{\ite}[0]{\ensuremath{\mathsf{it}}}
\newcommand{\append}[0]{\ensuremath{\mathsf{append}}}

\newcommand{\res}[0]{\ensuremath{\mathsf{res}}}
\newcommand{\CC}[0]{\ensuremath{\mathsf{C}}}
\newcommand{\LL}[0]{\ensuremath{\mathsf{L}}}

\newcommand{\SMTSolve}[0]{\ensuremath{\mathsf{SMTSolve}}}

\newcommand{\family}{\ensuremath{\mathsf{family}}}

\algnotext{EndFor}
\algnotext{EndIf}
\algnotext{EndWhile}

\newcommand{\citet}[1]{\cite{#1}}

\newtheorem{theorem}{Theorem}
\newtheorem{definition}{Definition}
\newcommand{\FixHash}{\ensuremath{\mathsf{FixLastHash}}}

\newcommand{\Hp}[1]{{h^{(#1)}}}

\newcommand\ol\overline
\newcommand\fct\rightarrow

\title{Approximate SMT Counting \\ Beyond Discrete Domains}

\author{\IEEEauthorblockN{Arijit Shaw}
\IEEEauthorblockA{\textit{Chennai Mathematical Institute, India}\\
\textit{IAI, TCG CREST, Kolkata, India}}
\and
\IEEEauthorblockN{Kuldeep S. Meel}
\IEEEauthorblockA{\textit{Georgia Institute of Technology, USA}\\
\textit{University of Toronto, Canada}}
}

\begin{document}
\maketitle
\bstctlcite{IEEEexample:BSTcontrol}

\begin{abstract}
  Satisfiability Modulo Theory (SMT) solvers have advanced automated reasoning, solving complex formulas across discrete and continuous domains. Recent progress in propositional model counting motivates extending SMT capabilities toward model counting, especially for hybrid SMT formulas. Existing approaches, like bit-blasting, are limited to discrete variables, highlighting the challenge of counting solutions projected onto the discrete domain in hybrid formulas.

  We introduce {\tool}, an SMT model counter for hybrid formulas that uses hashing-based approximate model counting to estimate solutions with theoretical guarantees. {\tool} makes a logarithmic number of SMT solver calls relative to the projection variables, leveraging optimized hash functions.  {\tool} achieves significant performance improvements over baselines on a large suite of benchmarks. In particular, out of  {\tot} instances, {\pact} successfully finished on {\mxp} instances, while the current state-of-the-art counter could only finish on  {\mxb} instances.{\footnote{A preliminary version of this paper appears at the Design Automation Conference (DAC) 2025. }}

\end{abstract}

\section{Introduction}\label{sec:intro}

Propositional model counting is the task of counting the number of satisfying assignments for a given Boolean formula. Recent advances in model counters  have made them useful for solving a variety of real-world problems~\cite{CMV21,GSS21} such as  probabilistic inference~\cite{CD08}, software verification~\cite{TW21}, network reliability~\cite{DMPV17}, and neural network verification~\cite{BSSM+19}. The development of model counters has been motivated by the success of SAT solvers over the past few decades, which allowed researchers to explore problems beyond mere satisfiability.

Concurrently, the success of SAT solvers~\cite{FHIJ21} led to an interest in solving the satisfiability of formulas where the variables are not just Boolean. This interest gave rise to the field of Satisfiability Modulo Theories (SMT), which includes a range of theories such as arithmetic, bitvectors, and data structures~\cite{KS16}. SMT theories, inspired by application needs, offer more succinct problem representations than Boolean satisfiability. There has been a significant development in the design of SMT solvers in recent years~\cite{boolector,mathsat5,BSST21,cvc5,bitwuzla}. The compactness of SMT and recent solver advancements have made them useful in software and hardware verification~\cite{MMBD+18,HJ20}, security~\cite{BBBB+20}, test-case generation, synthesis, planning~\cite{CMZ20}, and optimization~\cite{SSA16}.

In light of the availability of powerful SMT solvers, a natural next challenge that has emerged is exploration of techniques for model counting for SMT formulas.
A significant effort has been devoted to problems where the underlying theory is discrete,  such as bitvectors~\cite{CMMV16,KM18}, linear integers ~\cite{GMMZ+19,GB21,G24}, and strings~\cite{ABB15}. Recent studies have shown that in these cases, reducing the problem to Boolean model counting is often the most effective approach~\cite{SM24}.  There are, however, several applications that require reasoning over  both continuous and discrete variables, which we refer to as {\em hybrid SMT formulas}.

The hybrid SMT formulas are defined over both discrete and continuous variables, and we are interested in solutions projected over discrete variables. Our investigations for the development of counting techniques for hybrid SMT formulas are motivated by their ability to model several interesting and relevant applications, such as robustness quantification of cyber-physical systems and counting reachable paths in software (for detailed discussion, see Section~\ref{sec:applications}). The most relevant work to our setting is that of Chistikov, Dimitrova, and Majumdar~\cite{CDM17}, who introduced the notion of \emph{measured theories}, which provides a unifying definition for quantitative problems across discrete and continuous SMT theories, as well as their combinations.  Building on this framework, they developed hashing-based approximation algorithms for model counting over hybrid SMT formulas. Their techniques, however, face scalability hurdles, as discussed in Section~\ref{sec:results} and consequently, the development of scalability techniques for hybrid SMT formulas has remained a formidable challenge.

The primary contribution of this work is an affirmative answer to the aforementioned challenge via the development of a model counting tool, {\tool}, for efficient projected counting of hybrid SMT formulas. The {\tool} approximates the model count with $(\varepsilon, \delta)$ guarantees. {\pact} supports SMT formulas with various theories, including linear and non-linear real numbers, floating-point arithmetic, arrays, bit-vectors, or any combination thereof. The projection variables are over bit-vectors. {\pact} employs a hashing-based approximate model counting technique, utilizing various hash functions such as multiply-mod-prime, multiply-shift, and XOR.
The algorithm makes $\mathcal{O}(\log(|S|))$ calls to the SMT oracle, where $S$ is the set of projection variables. We have implemented a user-friendly open-source tool based on {\cvc}.  {\tool} supports a diverse array of theories including  QF\_ABV, QF\_BVFP,    QF\_UFBV, QF\_ABVFPLRA,    QF\_ABVFP,       QF\_BVFPLRA.

To demonstrate runtime efficiency, we conduct an extensive empirical evaluation over {\tot} benchmark instances. On these instances, {\pact} successfully finished on {\mxp} instances, while the current state-of-the-art~\citet{CDM17} could finish only on   {\mxb} instances.

\subsection{Applications} \label{sec:applications}

We now discuss four motivating applications for counting over hybrid SMT formulas.

\mparagraph{Robustness Analysis of Automotive Cyber-Physical Systems.}
Evaluating robustness is crucial in automotive cyber-physical systems (CPS), especially with the rise of autonomous vehicles. Koley et al.~\citet{KDM+23} encoded the problem using SMT to identify potential CPS attack vectors, incorporating both discrete and continuous variables to represent cybernetic and physical aspects, respectively. This framework extends to a quantitative approach, where the problem becomes an SMT counting query. Robustness is assessed by counting potential attack points, with the projection set defined by the system’s input parameters.

\mparagraph{Reachability Analysis of Critical Software.} Consider a control-flow graph (CFG) of critical software, where we are interested in knowing how many different paths exist in that CFG, such that some violating conditions are reached. We can encode this problem as a counting problem on the SMT formula with discrete and continuous variables, and the projection set would contain Boolean variables indicating whether a node of CFG is reachable. The projected model count will give the number of satisfying paths in CFG. %

\mparagraph{Quantitative Software Verification.}
To ensure software reliability, identifying bugs is not always sufficient; a quantitative approach is vital for understanding their impact. A program with an assertion is converted into an SMT formula through a Single Static Assignment (SSA), revealing inputs that lead to assertion failures by counting these specific inputs. Teuber and Weigl\citet{TW21} reduced the quantitative verification to projected counting over hybrid SMT formulas, wherein the underlying theory is  {QF\_BVFP}.

\mparagraph{Quantification of Information Flow.}
In the domain of software reliability, the quantification of information flow represents a critical challenge, particularly in measuring information leakage within industrial software applications. Phan and Malacaria~\citet{PM14} showed that the problem of quantification of information flow in the case of standard programs can be reduced to the task of counting over hybrid SMT formulas defined over {QF\_BVFP}.

\subsection{Organization}
The rest of the paper is organized as follows: We introduce the preliminaries and related work in \cref{sec:background}. In \cref{sec:framework}, we present an overview of our framework, {\tool}. We describe our experimental methodology and results in \cref{sec:results}. Finally, we conclude in \cref{sec:concl}.

\section{Preliminaries}
\label{sec:background}

\mparagraph{Satisfiability Modulo Theory (SMT)}~\cite{KS16} combines Boolean satisfiability (SAT) with theories such as integer and real arithmetic, bit-vectors, arrays, enabling efficient and automated analysis of logical formulas involving various data types. \textit{SMT solvers} solves the satisfiability of an SMT formula.

\mparagraph{Hybrid SMT Formulas.} Some SMT theories are discrete (such as bitvectors and integers), while others are continuous (such as reals and floating points). 
We define a \textit{hybrid SMT formula} as an SMT formula that combines two or more theories, where there is at least one discrete theory and one continuous theory. For example, a formula in {QF\_BVLRA} is a hybrid formula because it contains both real variables (continuous) and bitvector variables (discrete). %

\mparagraph{Projection Set and Projected Solutions.}
Let $F$ represent an SMT formula, where $\vars{F}$ signifies the set of all variables of $F$. A \textit{projection set} $S$ is a subset of $\vars{F}$. Given an assignment $\tau$ to $\vars{F}$,  $T_{\downarrow S}$ denotes the projection of $\tau$ on $S$ 
$\sol(F)$ denotes the set of all solutions to the formula $F$. $\solproj{F}{S}$ represents the set of all solutions of $F$ projected on $S$. In the context of this paper, $S$ is a set of discrete variables and therefore, $\solproj{F}{S}$ is a finite set. 

\mparagraph{Model Counting.}
Given a formula $F$ and a projection set $S$, the problem of model counting is to compute $|\solproj{F}{S}|$.
An  \emph{approximate model counter} takes in a formula $F$, projection set $S$, tolerance parameter $\varepsilon$, and confidence parameter $\delta$, and returns $c$ such that $\Pr\left[\frac{|\solproj{F}{S}|}{1+\varepsilon} \leq c \leq (1+\varepsilon) |\solproj{F}{S}| \right] \geq 1-\delta$.

\mparagraph{Hash functions.}
A \textit{hash function} $h: U \to [m]$ maps elements from a universe $U$ to a range $[m] = \{0, 1, \dots, m-1\}$. A \textit{pairwise independent hash function} is a hash function chosen from a \textit{family} $\mathcal{H}$ of functions $h: U \to [m]$ such that, for any two distinct elements $x_1, x_2 \in U$ and for any $i_1, i_2 \in [m]$: $\Pr[h(x_1) = i_1 \land h(x_2) = i_2] = \sfrac{1}{m^2}.$
A \textit{vector hash function} $h: U^d \to [m]$ extends this concept by mapping $d$-dimensional vectors of $w$-bit integers to the range $[m]$.
In this paper, the term \textit{hash function} refers to \textit{hash-based constraint}, represented as $h(\mathbf{x}) = \alpha$. The solutions of the formula $F \land (h(\mathbf{x}) = \alpha)$ form a subset of the solutions to $F$, restricted to those where the hash function maps to $\alpha$.

\subsection{Problem Statement}
We shall introduce the projected counting problem, defined by the specific theory of the formula and the projection variables.

\begin{definition}[\ctp(F,S)]
  Given a logical formula $F$ defined over the SMT theory $\mathcal{T} \cup \mathcal{P}$; and a projection set $S$ on theory $\mathcal{P}$; where $\mathcal{T}$ is either discrete or continuous or combination of both, and $\mathcal{P}$ is a discrete theory, {\ctp$\mathsf{(F,S)}$} refers to the problem of counting $|\solproj{F}{S}|$.
\end{definition}
In this work, we consider $\mathsf{BV}$ as $\mathcal{P}$. Any possible theory or combination of theories can serve as $\mathcal{T}$. Consequently, the resulting counting problems take forms such as $\cfplrabv$, $\cbvfp$, and similar variations.

\subsection{Related Work}

The success of propositional model counters, particularly approximate model counters, prompted efforts to extend the techniques to word-level constraints.
The most relevant work in this direction is due to Chistikov, Dimitrova, and Majumdar~\cite{CDM17}. Their algorithm supports projected counting over hybrid domains by treating the continuous variables as existentially quantified. They reduce model counting to a sequence of SMT queries. To obtain an approximation with a desired precision, these SMT queries contain multiple copies of the original SMT formula, and the hashing constraints are applied to the duplicated free variables. %

In the following years, Chakraborty et al.\citet{CMMV16} designed {\smap} for counting over theory of bit-vector by lifting the hash functions for word-level constraints. Kim and McCamant\citet{KM18} designed a system to estimate model count of bitvector formulas.  Ge et al.~\citet{GMLT18} developed a probabilistic polynomial-time model counting algorithm for bit-vector problems, and also developed a series of algorithms in the context of related SMT theories to compute or estimate~\cite{GB21,GMZ18,GMMZ+19} the number of solutions for linear integer arithmetic constraints. %

A closely related problem in the hybrid domain of Boolean and rational variables is \textit{Weighted Model Integration} (WMI)~\citet{BPdB15}, which involves computing the volume given the weight density over the entire domain. Extensive research addresses WMI through techniques such as predicate abstraction and All-SMT~\citet{MPS17,MPS19}, as well as methods leveraging knowledge compilation~\citet{KMSB+18}.

Hashing-based approximate model counting has been extensively studied over the past decades~\cite{AHT18,AT17,BdBP15,CMMV16,CMV13,CMV16,GSS06,SM19,S83,YM23,ZCSE16}. Chakraborty et al.\cite{CMV21} showed that variations in a few key components in a generalized framework account for the diversity in prior approaches. While prior works focused on discrete domains such as Boolean variables\cite{CMV13,CMV16,YM23} and bitvectors~\cite{CMMV16,CDM17}, our approach extends this framework to hybrid SMT formulas.

\section{Algorithm and Implementation}
\label{sec:framework}

We introduce {\tool}, our tool for approximate counting of SMT formulas. It processes a formula $F$, a set of projection variables $S$, a \textit{tolerance} $\varepsilon$, and a \textit{confidence} $\delta$ to produce an approximation of $|\solproj{F}{S}|$ within the desired tolerance and confidence. The main idea behind {\tool} involves dividing the solution space into equally sized \textit{cells} using hash functions and then enumerating the solutions within each \textit{cell}.

\begin{algorithm}[!htbp]
  \caption{{\tool}($F, S, \varepsilon, \delta, \family$)}
  \label{alg:pact}
  \begin{algorithmic}[1]
    \State $\LL \leftarrow\varnothing, \ite \gets 0$\label{algline:init}

    \State $\thresh, \itercount, \ell \gets  \getconstants (\varepsilon, \delta, \family) $\label{algline:thresh}

    \State $\CC[0] \gets \Call{\saturatingcounter}{F,S, \thresh}$\label{algline:saturating1}
    \If{$\CC[0] \neq \mathsf{T}$} \Return $\CC[0]$ \label{algline:checksaturating1}
    \EndIf

    \While{$\ite < \itercount$}\label{algline:repeat}
    \State $\CC \gets \varnothing, \iterdone \gets \bot, i \gets 0$
    \State $H \gets \generatehash (S, \ell, \family)$ \label{algline:genhash}
    \While {$\iterdone = \bot$} \label{algline:innerloopstart}
    \State $i \gets \nextIndex(\CC, i)$\label{algline:nextindex}
    \State $\CC[i] \gets \Call{\saturatingcounter}{F \land H_{[i]}, P, \thresh}$\label{algline:saturating}
    \If{$\CC[i] < \thresh \land \CC[i-1] = \top$}\label{algline:saturatingif}
    \State $\CC', H' \gets \FixHash(F,S, \CC , H,i, \ell)$ \label{algline:fixhash}
    \State $\LL.\append (\GetCount(\CC'[i],H'))$ \label{algline:append}
    \State $\iterdone \gets \top, \ite$++ \label{algline:loopend}
    \EndIf
    \EndWhile
    \EndWhile
    \State \Return $\FindMedian(\LL)$
  \end{algorithmic}
\end{algorithm}

\Cref{alg:pact} presents the main algorithm {\pact}. The algorithm starts by setting constants of the algorithm, the value for {\thresh} and {\itercount} from the values of $\varepsilon$ and $\delta$, depending on the hash family being used. The constants arise from technical calculations in the correctness proof of the algorithm, and the values are shown {\getconstants} subroutine (\Cref{alg:getconstants}). The value of {\thresh} determines the maximum size of a \textit{cell}. A cell is considered \textit{small} if it has a number of solutions less than this threshold. The value of {\itercount} determines how many times the main loop of the program (lines \ref{algline:repeat} - \ref{algline:loopend})  is repeated. In each iteration of the main loop, an approximate count is generated, which is stored in the list {\LL}. While each of the approximate counts might fail to provide an estimate with the desired $\delta$, the median of the counts of this list gives the approximation of model count with $(\varepsilon,\delta)$ guarantees.

The main loop of the algorithm begins with the subroutine {\generatehash}, which produces a list of hash functions $H$ selected from one of the families $\hshift$, $\hprime$, or $\hxor$ to be used during the current iteration.
In each iteration, {\pact} maintains a list $\CC$ of numbers, where the element $\CC[i]$ represents the size of a \textit{cell} after applying the first $i$ hash functions from $H$, denoted as $H_{[i]}$. The subroutine {\nextIndex}, called in line~\ref{algline:nextindex}, uses a galloping search to identify an index $i$ in $\CC$ where the value of $\CC[i]$ has been computed. The parameter $\ell$ in {\generatehash} determines the range of hash functions generated. Specifically, (i) for {\hshift}, the range of hash functions is set to $2^\ell$, and (ii) for {\hprime}, {\generatehash} constructs hash functions of a range of the smallest prime larger than $2^\ell$.

Following that, in line~\ref{algline:saturating}-\ref{algline:saturatingif}, {\pact} checks by a call to {\saturatingcounter} whether $|\solproj{F \land H_{[i]}}{S}| < \thresh$ and $|\solproj{F \land H_{[i]}}{S}| \geq \thresh$. With this condition, {\pact} is enabled a rough estimate of the model count, but to get the estimate within desired error bounds, {\pact}  uses the {\FixHash} subroutine in line~\ref{algline:fixhash}.
This subroutine eliminates the last hash, $H[i]$, and introduces a new hash function that reduces the number of solution partitions.
The subroutine iterates the procedure to find two hash functions $h'$ and $h''$, such that $h''$ divides the space into $ \sfrac{k}{2} $ parts, while $h'$ divides into $k$ parts. Now if  $|\solproj{F \land H_{[i-1]} \land h'}{S}|  < thresh, |\solproj{F \land H_{[i-1]} \land h''}{S}| \geq \thresh$, {\FixHash} returns $ H' = H_{[i-1]} \land h'$ and $C' = |\solproj{F \land H_{[i-1]} \land h'}{S}|$.
However, when {\pact} uses {\hxor}, the call to {\FixHash} is unnecessary, as it already partitions the solution space into two parts using one hash function.

The subroutine {\GetCount} approximates $\solproj{F}{S}$ by multiplying $C'$ with the number of partitions generated by all the hashes used in $H'$. This number is then appended to the list {\LL}, and {\pact} continues to the next iteration of the main loop. Once the main loop generates count for {\itercount} times, we take the median of all the counts, which is an approximation for $\solproj{F}{S}$ with desired guarantees.

\begin{algorithm}[htb]
  \caption{$\FixHash(F,S,\CC, H,i, \ell)$}
  \label{alg:getstatus}
  \begin{algorithmic}[1]
    \If{\family = \hxor}
    \Return $\CC, H$
    \EndIf
    \While{$\ell > 1$}
    \State   $\ell \gets \floor{\sfrac{\ell}{2}}$
    \State $\Hp{\ell} \gets \generatehash (S, \ell, \family)$ \label{algline:genhash}
    \State $c \gets \Call{\saturatingcounter}{F \land H_{[i-1]} \land \Hp{\ell}, S, \thresh}$
    \If{$c \neq \top$}
    $\CC[i] = c, H[i] = h^{(\ell)}$
    \Else{}
    \Return $\CC,H$
    \EndIf
    \EndWhile
    \State \Return $\bot$
  \end{algorithmic}
\end{algorithm}

\begin{algorithm}
  \caption{\getconstants($\varepsilon$, $\delta$, $\family$)}
  \label{alg:getconstants}
  \begin{algorithmic}[1]
    \State $\thresh \gets  1+9.84 \left( 1 + \frac{\varepsilon}{1+\varepsilon} \right) \left( 1 + \frac{1}{\varepsilon} \right)^2 $\label{algline:thresh}
    \If{$\family = \hxor$}
    $\iters \gets \lceil 17 \log \frac{3}{\delta} \rceil, \ell \gets 1$
    \Else{}
    $\iters \gets \lceil 23 \log \frac{3}{\delta} \rceil, \ell \gets 4$
    \EndIf
    \State \Return $\thresh, \iters, \ell$
  \end{algorithmic}
\end{algorithm}

\subsection{Hash functions.}
The choice of hash function is one of the most important parts in a hashing-based model counter like {\pact}.
In {\pact}, we experiment with three different pairwise independent vector hash functions, which have been used in different literature.
\begin{itemize}
  \item \textit{Multiply mod prime} (\hprime)~\cite{T19}: For a prime number $p$, and independent random values $\mathbf{a} = a_0,\ldots,a_{d-1},b \in [p]$,  the hash function from $[p]^d$ to $[p]$ is given by:
        $$h_{\mathbf{a},b}(\mathbf{x})\equiv\left(\sum_{i\in [d]} a_i x_{i}  +b\right)\bmod p = \alpha$$

  \item \textit{Multiply-shift (\hshift)} ~\cite{D96}: For independent random values $\mathbf{a}= a_0,\ldots,a_{d-1},b \in [2^{\overline{w}}]$, where $\ol{w} \geq w + \ell -1$ the hash function from $[2^w]^d$ to $[2^\ell]$ is given by:
        $$ h_{\mathbf{a},b}(\mathbf{x})\equiv     \left(\sum_{i\in [d]}a_ix_i+b\right)[\overline{w}-\ell,\overline{w})= \alpha $$

  \item \textit{Bitwise XOR (\hxor)}~\cite{CW77}: This is a particular case of the multiply mod prime scheme when $p = 2$.
        $$ h_{\mathbf{a}}(\mathbf{x}) \equiv  \bigoplus_{\{i | a_i = 1\}} x_{i} = \alpha$$

\end{itemize}

Given the type of hash function to generate, the procedure {\generatehash} generates the constraints of the form $h_{\mathbf{a},b} = \alpha$, where $\alpha$ is a randomly chosen element from the domain of the hash function.
When this hash function-based constraint $H$ is added to the input formula $F$, $F\land H$ satisfies only those solutions of $H$ for which the hashed value by $H_{\mathbf{a},b}$ is $\alpha$.
therefore, the number of solutions of $F \land H$ is approximately $p^{th}$ fraction of number of solution of $F$.
  {\hprime} and {\hshift} are word-level hash functions that allow us to choose a number of partitions we want to divide our solution space into. When {\pact} uses these hash functions, along with the projection set $S$, and the {\family}, {\generatehash} takes a parameter $\ell$, and generates a hash function with domain size $p$, such that $2^\ell\leq p <2^{\ell+1}$. Specifically, in {\hshift}, $p = 2^\ell$, in {\hprime}, $p$ is the smallest prime $ >2^\ell$.
In case of {\hxor}, $\ell =1, p = 2$.
As $S$ is evident from the context, and $\mathbf{a},b$ are randomly generated, we use the notation $h$ instead of $h(S)$ to denote the hash functions.

While each of the hashing constraints partitions the solution space into $p$ slices, we often want to partition it into more.
To partition the solution space into $p^c$ cells, we use the Cartesian product of $c$ hash functions: $\HH \times \HH \times \dots \times \HH$. We use the notation $H_{[i]}$ to denote the Cartesian product of the first $i+1$ hash functions, i.e., $H_{[i]} = h_0 \times h_1 \times \dots \times h_i$.

\mparagraph{Slicing.}
In {\generatehash} subroutine of {\pact}, the hash functions have particular domain sizes. But, the bitvectors can have arbitrary width; we \textit{slice} them into bitvectors of smaller width so that the value of (sliced) bitvector lies within the domain of the hash function.
Instead of defining hash functions on the variables from $S$, we define hash functions on \textit{slices} of the variables, defined as follows:
For a bitvector $\mathbf{x}$ of width $w$, we define $\ceil{w/\ell}$ slices of width $\ell$:
$x(\ceil{w/\ell}-1), \dots, x(1), x(0)$, where $x(i) = x\left[{(i+1)\ell}-1 : {i\ell}\right]$.

\subsection{Enumerating in a cell.} Another crucial step in {\pact} is to determine when the size of a cell is less than the threshold. {\pact} uses the  {\saturatingcounter} subroutine to enumerate solutions of the formula $F\land H_{[i]}$ to determine the size. An SMT solver is employed to find a solution $\res$ for the given formula $F$. The projection of this solution, $\res_{\downarrow S}$, is then \textit{blocked} by adding a constraint $\neg (\res_{\downarrow S})$. Subsequently, the solver is asked to find another solution. This process is repeated in a loop until the solver finds {\thresh} many solutions or reports UNSAT.

\subsection{Searching for the number of partitions.} The number of index of {\CC} goes upto \OO{|S|}. The task for {\nextIndex} subroutine is to find the index, for which the cell size is in the range $[1,\thresh]$. The {\nextIndex} finds the index by \OO{\log(|S|)} many calls by employing a galloping search method.%

\vspace{.5cm}

We defer the detailed descriptions of {\generatehash}, {\nextIndex}, {\saturatingcounter},  {\GetCount} and to the technical report of the paper.

\subsection{Analysis}

\begin{theorem}
  Let $F$ be a formula defined over a set of variables $V$ and discrete projection variables $S$. Let $\est = \tool(F,S,\varepsilon,\delta)$ be the approximation returned by {\tool}, and $c = |\solproj{F}{S}|$. Then, $$\Pr \left[\frac{c}{\ope} \leq \est \leq (\ope) c \right] \geq 1- \delta $$
  Moreover, {\tool} makes $\mathcal{O}(\log (|S|)\frac{1}{\varepsilon^2}\log( \frac{1}{\delta})  )$ many calls to an SMT solver.
\end{theorem}

\begin{proof}
  We defer the proof to the accompanying technical report.
\end{proof}

\subsection{Impact of Hash Function Families}
Our empirical evaluation indicates {\saturatingcounter} is computationally the most expensive subroutine during execution of {\pact}. Recall that  $\saturatingcounter$ is invoked over the formula instance $F\land H_{[i]}$; naturally, the choice of the hash function $\family$ impacts the practical difficulty of the instance  $F\land H_{[i]}$. Below, we highlight the tradeoffs offered by different hash function families along many dimensions.
\begin{itemize}

  \item \textit{Bit-level vs. Bitvector Operations:} The hash function {\hxor} operates at the bit level, whereas {\hshift} and {\hprime} function on bitvectors, making the latter two more amenable to SMT reasoning. This fundamental difference in operation also enables {\hprime} and {\hshift} to vary the hash function's domain sizes, a flexibility not available with \hxor.

  \item \textit{Number of hash functions:} As each hash function cannot partition the solution space by more than two partitions, {\hxor} requires more number of constraints for counting an instance compared to {\hshift} and \hprime. A number of constraints generally make the problem harder for the solver to solve efficiently.

  \item \textit{Complexity of required operations:} Multiplication and modulus operations impose significant computational overhead on SMT solvers, making {\hprime} and {\hshift} more complicated than {\hxor}. Moreover, when {\pact} uses {\hxor}, it leverages the native XOR reasoning capability of $\mathsf{CryptoMiniSat}$ SAT solver inside {\pact}, further increasing the counter's performance.

  \item \textit{Requirement of bitwidth:} To represent the value of $\sum a_i x_i$, a bitvector of width $2w$ is required in {\hshift}, whereas a bitwidth of $2w + d$ is needed in {\hprime}, as the hash value in {\hshift} is calculated modulo $2^{2w}$, while in {\hprime} is computed modulo a prime. Since SMT solver performance degrades quickly with increasing bitwidth, {\hshift} is a more favorable choice in this context. In an alternative implementation of {\hprime}, each term $a_i x_i$ could be represented modulo a prime, but this would require an additional $d$ modulo operations, increasing complexity.
\end{itemize}

\subsection{Implementation and Supported Theories}
We implement {\tool} on top of {\cvc}, a modern SAT solver. {\pact} uses {\cvc} for parsing the formula and running the {\SMTSolve} procedure. In the problem of {\ctp},  {\pact} solves all of the theories and theory combinations in SMT-Lib for {\TT} and theories of bit-vectors for {\PP}. To enhance efficiency, we utilize the SMT solver in its incremental mode, allowing each subsequent query to leverage the information gained from previous calls. Similar incremental calls are different calls to {\saturatingcounter} at different iterations of {\tool}.

\begin{table}[t]
  \caption{Number of instances counted. (Projection on BV variables.)}%
  \centering
\begin{tabular}{llccc}
\toprule
Logic & {\CDM} & {\pprime} & {\pshift} & {\pxor} \\
\midrule
QF\_ABVFPLRA (19) & $-$ & $-$ & $-$ & $1$ \\
QF\_ABVFP (13) & $-$ & $1$ & $1$ & $7$ \\
QF\_ABV (2713) & $11$ & $-$ & $-$ & $284$ \\
QF\_BVFPLRA (46) & $-$ & $-$ & $-$ & $30$ \\
QF\_BVFP (129) & $71$ & $23$ & $37$ & $117$ \\
QF\_UFBV (199) & $1$ & $9$ & $2$ & $17$ \\ \midrule
Total (3119) & $83$ & $33$ & $40$ & $456$ \\
\bottomrule
\end{tabular}

  \label{tab:solved}
\end{table}

\section{Experimental Evaluation}
\label{sec:results}

We implemented{\footnote{The source code is available at  \url{https://github.com/meelgroup/pact}.  }} the proposed algorithm on top of {\sota} SMT solver {\cvc}. We used the following experimental setup  in the evaluation:

\mparagraph{Baseline.} We compared the performance of {\pact} with the current state-of-the-art tool by  Chistikov, Dimitrova, and Majumdar~\citet{CDM17}. We refer to the tool as {\CDM}, after the initials of its authors.

\mparagraph{Benchmarks.}
Our benchmark suite comprises {\tot} instances from the SMT-Lib 2023 release. We adopted a benchmark selection methodology inspired by early works on propositional model counting. We initially selected all instances supported by six theories. Subsequently, we filtered out instances where the number of solutions was very small (less than 500 models) or where  even satisfiability was computationally challenging, as determined by {\cvc}'s inability to find a satisfying assignment within 5 seconds. To avoid overweighting multiple near-identical benchmarks, we also sample to at most five benchmarks per \textit{cluster} (benchmarks whose file names and  differ only in index-level details).

\mparagraph{Environment.} We conducted all our experiments on a high-performance computer cluster, with each node consisting of Intel Xeon Gold 6148 CPUs. We allocated one CPU core and an 8GB  memory limit to each solver instance pair. To adhere to the standard timeout used in model counting competitions, we set the timeout for all experiments to 3600 seconds. We use values of $\varepsilon=0.8$ and $\delta=0.2$, in line with prior work in the model counting community.

\noindent
We conduct extensive experiments to understand the following:
\begin{enumerate}[font=\bfseries RQ, leftmargin=\widthof{[RQdd]}+\labelsep]
  \item How does the runtime performance of {\tool} compare to that of {\CDM}, and how does the performance vary with different hash function families?
  \item How accurate is the count computed by {\tool} in comparison to the exact count?

\end{enumerate}

\mparagraph{Summary of Results.}
{\tool} solves a significant number of instances from the benchmarks. It counted {\mxp} instances, while the current state-of-the-art counted only {\mxb} instances. Among different hash families, {\pact} performs the best while it uses {\hxor} hash functions. The accuracy of {\tool} is also noteworthy; the average approximation error is 3.3\% while using {\hxor} hashes\footnote{The benchmarks and logfiles are available at \url{https://doi.org/10.5281/zenodo.16413909}}.

\begin{figure}[!tb]
  \centering
  \resizebox{.95\linewidth}{!}{\input{figures/cactus.pgf}}
  \caption{Cactus plot comparing performances of {\tool} and {\CDM}.}
  \label{fig:cactus}
\end{figure}

\subsection{Performance of {\tool}}

\label{subsec:performance}
We evaluate the performance of {\tool} based on two metrics: the number of instances solved and the time taken to solve those instances. To differentiate {\pact} utilizing different hash function families, we use the notations {\pprime}, {\pshift}, and {\pxor}.

\mparagraph{Instances solved.} For each of the logic, we look at the number of instances solved. Out of {\tot} instances,  {\CDM} could solve only {\mxb} instances. Conversely, {\pxor} could solve {\mxp} instances,  demonstrating a substantial improvement compared to {\CDM}. The performance varies across different logics, which we represent in \Cref{tab:solved}. The number of instances solved is $14.6\%$ of the total number of instances, which is expected, given the target problem for these instances is satisfiability, and {\pact} solves a more complex problem.

\mparagraph{Comparison of Hash Functions.}
In the \Cref{tab:solved}, we compare the performance of {\tool} when it utilizes different hash functions for partitioning the solution space. The best performance is shown by {\pxor}, which solved {\mxp} instances. The performances of {\pprime} and {\pshift} are similar,  solving {\mxpp} and {\mxps} instances.%

\mparagraph{Solving time comparison.} A performance evaluation of {\CDM} and {\tool} is depicted in \Cref{fig:cactus}, which is a cactus plot comparing the solving time. The $x$-axis represents the number of instances, while the $y$-axis shows the time taken. A point $(i, j)$ in the plot represents that a solver solved $j$ benchmarks out of the {\tot} benchmarks in the test suite in less than or equal to $j$ seconds. The different curves plot the performance of {\CDM} and {\pact} with different hash functions.

\begin{figure}[!tb]
  \centering
  \resizebox{0.83\linewidth}{!}{\input{figures/correctness.pgf}}
  \caption{Accuracy check: observed error in {\pact} vs. the theoretical bound.}
  \label{fig:count-comp}
\end{figure}

\subsection{Quality of Approximation}
As none of the existing tool returns the exact count, to get the exact number of solutions, we developed an enumeration-based counter, referred to as {\enum}. Here, we use the state-of-the-art  SMT solver, {\cvc}, the same solver employed by {\pact}.
From our benchmark set, only 9 instances were solved by {\enum}. To increase the number of instances for which we know the exact count, we also include the benchmarks with model counts between 100 and 500 in this section of the paper - resulting in {\mxba} instances. We quantify the quality of approximation with the parameter error $e = \max\left(\frac{b}{s},\frac{s}{b}\right) -1$, where $b$ is the count from {\enum} and $s$ from {\tool}. this definition of error aligns with the $\varepsilon$ used in the algorithm's theoretical guarantees and can be interpreted as the observed value of $\varepsilon$. Analysis of all {\mxba} cases found that for {\pxor}, the maximum $e$ to be $0.26$ and the average to be $0.03$,   signifying {\tool} substantially outperforms its theoretical bounds, which is 0.8. In \Cref{fig:count-comp}, we illustrate the quality of approximation for these instances. The $x$-axis lists the instances, while the $y$-axis displays the relative error exhibited by a configuration of {\pact}. A dot $(x,y)$ in the graph indicates $x^{th}$ instance showed a relative error of $y$. The graph indicates that, for most instances, the error lies below 0.2, with a few instances falling between 0.2 and 0.8. The error for {\pshift} and {\pprime} is relatively higher than {\pxor}, with average error being $0.07$ and $0.12$ and maximum error being $0.39$ and $0.48$.
Our findings underline {\tool}'s accuracy and potential as a dependable tool for various applications.

\section{Conclusion and Future Work}
\label{sec:concl}
In this work, we introduced {\pact}, a projected model counter designed for hybrid SMT formulas. Motivated by the diverse applications of model counting and the role of hashing, we explored the impact of various hash functions, examining both bit-level and bitvector-level approaches. Our empirical evaluation demonstrates that {\pact} achieves strong performance on a broad application benchmark set. A dedicated XOR reasoning engine significantly enhanced {\pact}'s performance, suggesting that further development of specialized reasoning engines for bit vector-level hash functions could be a promising research direction. Additionally, {\pact}'s theoretical framework supports the SMT theory of integers (with specified bounds) as projection variables; this feature is not yet implemented and remains an avenue for future work.

\section*{Acknowledgement}
We are thankful to Jaroslav Bendik, Supratik Chakraborty, Dmitry Chistikov, Rayna Dimitrova, Ashwin Karthikeyan, Aina Niemetz, Mathias Preiner, Uddalok Sarkar, Mate Soos and Jiong Yang for the many useful discussions.  This work was supported in part by the Natural Sciences and Engineering Research Council of Canada (NSERC) [RGPIN-2024-05956]. Part of the work was done when Arijit Shaw was a visiting graduate student at the University of Toronto. Computations were performed on the Niagara supercomputer at the SciNet HPC Consortium. SciNet is funded by Innovation, Science and Economic Development Canada; the Digital Research Alliance of Canada; the Ontario Research Fund: Research Excellence; and the University of Toronto.

\bibliographystyle{IEEEtran}

\bibliography{projsmt}

@ieeetranbstctl{IEEEexample:BSTcontrol,
  ctluse_forced_etal       = {yes},
  ctlmax_names_forced_etal = {6},
  ctlnames_show_etal       = {1},
  ctluse_url               = {no},
  ctldash_repeated_names   = {no}
}

@article{T19,
  title   = {High speed hashing for integers and strings},
  author  = {Thorup, Mikkel},
  journal = {arXiv preprint arXiv:1504.06804},
  year    = {2015}
}

@inproceedings{D96,
  title     = {Universal hashing and k-wise independent random variables via integer arithmetic without primes},
  author    = {Dietzfelbinger, Martin},
  booktitle = {Proc. of STACS},
  year      = {1996}
}

@inproceedings{HJ20,
  title     = {solc-verify: A modular verifier for solidity smart contracts},
  author    = {Hajdu, {\'A}kos and Jovanovi{\'c}, Dejan},
  booktitle = {Proc. of VSTTE},
  year      = {2020}
}

@inproceedings{MMBD+18,
  title     = {Cosa: Integrated verification for agile hardware design},
  author    = {Mattarei, Cristian and Mann, Makai and Barrett, Clark and Daly, Ross G and Huff, Dillon and Hanrahan, Pat},
  booktitle = {Proc. of FMCAD},
  year      = {2018}
}

@inproceedings{BBBB+20,
  title     = {Stratified abstraction of access control policies},
  author    = {Backes, John and Berrueco, Ulises and Bray, Tyler and Brim, Daniel and Cook, Byron and Gacek, Andrew and Jhala, Ranjit and Luckow, Kasper and McLaughlin, Sean and Menon, Madhav and others},
  booktitle = {Proc. of  CAV},
  year      = {2020}
}

@article{SSA16,
  title   = {Stochastic program optimization},
  author  = {Schkufza, Eric and Sharma, Rahul and Aiken, Alex},
  journal = {Communications of the ACM},
  volume  = {59},
  number  = {2},
  year    = {2016}
}

@article{CMZ20,
  title   = {Planning for hybrid systems via satisfiability modulo theories},
  author  = {Cashmore, Michael and Magazzeni, Daniele and Zehtabi, Parisa},
  journal = {Journal of Artificial Intelligence Research},
  volume  = {67},
  year    = {2020}
}

@inproceedings{SM24,
  title     = {{CSB: A Counting and Sampling Tool for Bitvectors}},
  author    = {Shaw, Arijit and Meel, Kuldeep S},
  booktitle = {Proc. of SMT Workshop at CAV},
  year      = {2024}
}

@inproceedings{CW77,
  title     = {Universal classes of hash functions},
  author    = {Carter, J Lawrence and Wegman, Mark N},
  booktitle = {Proc. of STOC},
  year      = {1977}
}

@inproceedings{PM14,
  title     = {Abstract model counting: a novel approach for quantification of information leaks},
  author    = {Phan, Quoc-Sang and Malacaria, Pasquale},
  booktitle = {Proc. of ASIACCS},
  year      = {2014}
}

@inproceedings{cvc5,
  title     = {cvc5: A versatile and industrial-strength SMT solver},
  author    = {Barbosa, Haniel and Barrett, Clark and Brain, Martin and Kremer, Gereon and Lachnitt, Hanna and Mann, Makai and Mohamed, Abdalrhman and Mohamed, Mudathir and Niemetz, Aina and N{\"o}tzli, Andres and others},
  booktitle = {Proc. of TACAS},
  year      = {2022}
}

@inproceedings{BPdB15,
  title     = {Probabilistic inference in hybrid domains by weighted model integration},
  author    = {Belle, Vaishak and Passerini, Andrea and Van den Broeck, Guy},
  booktitle = {Proc. of IJCAI},
  year      = {2015}
}

@inproceedings{BdBP15,
  title     = {Hashing-based approximate probabilistic inference in hybrid domains},
  author    = {Belle, Vaishak and Van den Broeck, Guy and Passerini, Andrea},
  booktitle = {Proc. of UAI},
  year      = {2015}
}

@inproceedings{MPS17,
  title     = {Efficient weighted model integration via SMT-based predicate abstraction},
  author    = {Morettin, Paolo and Passerini, Andrea and Sebastiani, Roberto},
  booktitle = {Proc. of AAAI},
  year      = {2017}
}

@inproceedings{KMSB+18,
  title     = {Efficient symbolic integration for probabilistic inference},
  author    = {Kolb, Samuel and Mladenov, Martin and Sanner, Scott and Belle, Vaishak and Kersting, Kristian},
  booktitle = {Proc. of IJCAI},
  year      = {2018}
}

@article{MPS19,
  title   = {Advanced SMT techniques for weighted model integration},
  author  = {Morettin, Paolo and Passerini, Andrea and Sebastiani, Roberto},
  journal = {Artificial Intelligence},
  year    = {2019}
}

@inproceedings{AHT18,
  title     = {Fast Sampling of Perfectly Uniform Satisfying
               Assignments},
  author    = {Achlioptas, Dimitris and Hammoudeh, Z. and Theodoropoulos,
               P.},
  booktitle = {Proc. of SAT},
  year      = {2018}
}

@inproceedings{boolector,
  title     = {{Boolector: An efficient SMT solver for bit-vectors and
               arrays}},
  author    = {Brummayer, Robert and Biere, Armin},
  booktitle = {Proc. of TACAS},
  year      = {2009}
}

@inproceedings{BSSM+19,
  title     = {Quantitative verification of neural networks and its
               security applications},
  author    = {Baluta, Teodora and Shen, Shiqi and Shinde, Shweta and
               Meel, Kuldeep S and Saxena, Prateek},
  booktitle = {Proc. of CCS},
  year      = {2019}
}

@incollection{BSST21,
  title     = {Satisfiability modulo theories},
  author    = {Barrett, Clark and Sebastiani, Roberto and Seshia, Sanjit
               A and Tinelli, Cesare},
  booktitle = {Handbook of satisfiability},
  year      = {2021},
  publisher = {IOS Press}
}

@inproceedings{mathsat5,
  title     = {{The mathsat5 SMT solver}},
  author    = {Cimatti, Alessandro and Griggio, Alberto and Schaafsma, Bastiaan Joost and Sebastiani, Roberto},
  booktitle = {Proc. of TACAS},
  year      = {2013}
}

@inproceedings{bitwuzla,
  title     = {Bitwuzla},
  author    = {Niemetz, Aina and Preiner, Mathias},
  booktitle = {Proc. of CAV},
  year      = {2023}
}

@article{CD08,
  title   = {On probabilistic inference by weighted model counting},
  author  = {Chavira, Mark and Darwiche, Adnan},
  journal = {Artificial Intelligence},
  volume  = {172},
  year    = {2008}
}

@article{CDM17,
  title={Approximate counting in SMT and value estimation for probabilistic programs},
  author={Chistikov, Dmitry and Dimitrova, Rayna and Majumdar, Rupak},
  journal={Acta Informatica},
  volume={54},
  number={8},
  pages={729--764},
  year={2017},
  publisher={Springer}
}

@inproceedings{CMMV16,
  title     = {Approximate probabilistic inference via word-level
               counting},
  author    = {Chakraborty, Supratik and Meel, Kuldeep and Mistry, Rakesh
               and Vardi, Moshe},
  booktitle = {Proc. of AAAI},
  year      = {2016}
}

@inproceedings{CMV13,
  title     = {A scalable approximate model counter},
  author    = {Chakraborty, Supratik and Meel, Kuldeep S and Vardi, Moshe
               Y},
  booktitle = {Proc. of CP},
  year      = {2013}
}

@inproceedings{CMV16,
  title     = {{Algorithmic Improvements in Approximate Counting for
               Probabilistic Inference: From Linear to Logarithmic SAT
               Calls.}},
  author    = {Chakraborty, Supratik and Meel, Kuldeep S and Vardi, Moshe
               Y},
  booktitle = {Proc. of IJCAI},
  year      = {2016}
}

@incollection{CMV21,
  title     = {Approximate model counting},
  author    = {Chakraborty, Supratik and Meel, Kuldeep S and Vardi, Moshe
               Y},
  booktitle = {Handbook of Satisfiability},
  year      = {2021},
  publisher = {IOS Press}
}

@inproceedings{DMPV17,
  title     = {Counting-based reliability estimation for
               power-transmission grids},
  author    = {Duenas-Osorio, Leonardo and Meel, Kuldeep and Paredes,
               Roger and Vardi, Moshe},
  booktitle = {Proc. of AAAI},
  year      = {2017}
}

@article{FHIJ21,
  title     = {SAT competition 2020},
  author    = {Froleyks, Nils and Heule, Marijn and Iser, Markus and
               J{\"a}rvisalo, Matti and Suda, Martin},
  journal   = {Artificial Intelligence},
  volume    = {301},
  year      = {2021},
  publisher = {Elsevier}
}

@inproceedings{GB21,
  title     = {Decomposition Strategies to Count Integer Solutions over
               Linear Constraints.},
  author    = {Ge, Cunjing and Biere, Armin},
  booktitle = {Proc. of IJCAI},
  year      = {2021}
}

@inproceedings{GMLT18,
  title     = {A New Probabilistic Algorithm for Approximate Model
               Counting},
  author    = {Ge, Cunjing and Ma, Feifei and Liu, Tian and Zhang, Jian
               and Ma, Xutong},
  booktitle = {Proc. of IJCAR},
  year      = {2018}
}

@inproceedings{GMMZ+19,
  title     = {Approximating integer solution counting via space
               quantification for linear constraints},
  author    = {Ge, Cunjing and Ma, Feifei and Ma, Xutong and Zhang, Fan
               and Huang, Pei and Zhang, Jian},
  booktitle = {Proc. of IJCAI},
  year      = {2019}
}

@inproceedings{GMZ18,
  author    = {Cunjing Ge and Feifei Ma and Jian Zhang},
  title     = {{VolCE: An Efficient Tool for Solving {\#}SMT(LA)
               Problems}},
  booktitle = {Proc. of PRUV Workshop at IJCAR},
  year      = {2018}
}

@inproceedings{G24,
  title     = {Approximate Integer Solution Counts over Linear Arithmetic Constraints},
  author    = {Ge, Cunjing},
  booktitle = {Proc. of AAAI},
  year      = {2024}
}

@inproceedings{ABB15,
  title     = {Automata-based model counting for string constraints},
  author    = {Aydin, Abdulbaki and Bang, Lucas and Bultan, Tevfik},
  booktitle = {Proc. of CAV},
  year      = {2015}
}

@incollection{GSS21,
  title     = {Model counting},
  author    = {Gomes, Carla P and Sabharwal, Ashish and Selman, Bart},
  booktitle = {Handbook of satisfiability},
  year      = {2021},
  publisher = {IOS press}
}

@inproceedings{KM18,
  title     = {Bit-vector model counting using statistical estimation},
  author    = {Kim, Seonmo and McCamant, Stephen},
  booktitle = {Proc. of TACAS},
  year      = {2018}
}

@book{KS16,
  title     = {Decision procedures},
  author    = {Kroening, Daniel and Strichman, Ofer},
  year      = {2016},
  publisher = {Springer}
}

@inproceedings{SM19,
  title     = {{BIRD: engineering an efficient CNF-XOR SAT solver and its
               applications to approximate model counting}},
  author    = {Soos, Mate and Meel, Kuldeep S},
  booktitle = {Proc. of AAAI},
  year      = {2019}
}

@inproceedings{S83,
  title     = {The complexity of approximate counting},
  author    = {Stockmeyer, Larry},
  booktitle = {Proc. of STOC},
  year      = {1983}
}

@inproceedings{GSS06,
  title     = {Model counting: A new strategy for obtaining good bounds},
  author    = {Gomes, Carla P and Sabharwal, Ashish and Selman, Bart},
  booktitle = {Proc. of AAAI},
  year      = {2006}
}

@inproceedings{AT17,
  title     = {Probabilistic model counting with short XORs},
  author    = {Achlioptas, Dimitris and Theodoropoulos, Panos},
  booktitle = {Proc. of SAT},
  year      = {2017}
}

@inproceedings{ZCSE16,
  title     = {Closing the gap between short and long xors for model counting},
  author    = {Zhao, Shengjia and Chaturapruek, Sorathan and Sabharwal, Ashish and Ermon, Stefano},
  booktitle = {Proc. of AAAI},
  year      = {2016}
}

@article{KDM+23,
  title   = {{CAD Support for Security and Robustness Analysis of Safety-critical Automotive Software}},
  author  = {Koley, Ipsita and Dey, Soumyajit and Mukhopadhyay, Debdeep and Singh, Sachin and Lokesh, Lavanya and Ghotgalkar, Shantaram Vishwanath},
  journal = {ACM Transactions on Cyber-Physical Systems},
  year    = {2023}
}

@inproceedings{TW21,
  author    = {Teuber, Samuel and Weigl, Alexander},
  title     = {Quantifying Software Reliability via Model-Counting},
  booktitle = {Proc. of QEST},
  year      = {2021}
}

@inproceedings{YM23,
  author    = {Yang, Jiong and Meel, Kuldeep S},
  booktitle = {Proc. of CAV},
  title     = {{Rounding Meets Approximate Model Counting}},
  year      = {2023}
}

\end{document}